\documentclass[submission,copyright,creativecommons]{eptcs}
\providecommand{\event}{LSFA 2026}

\usepackage{iftex}

\usepackage{amssymb}
\usepackage{amsmath}
\usepackage{amsthm}

\usepackage{bbm}
\usepackage{url}
\usepackage{hyperref}

\usepackage{ wasysym }
\usepackage{ dsfont }
\usepackage{wasysym}

\usepackage{tikz}
\usetikzlibrary{shapes,arrows}

\usepackage{graphicx}

\usepackage{xcolor}

\usepackage{amssymb}
\usepackage{bbm}
\usepackage{url}
\usepackage{hyperref}
\usepackage{scalerel}
\usepackage{mathtools}
\usepackage{csquotes}

\newcommand{\R}{\mathcal{R}}
\newcommand{\Tau}{\mathcal{T}}

\newcommand{\D}{\mathcal{D}}

\newcommand{\B}{\mathcal{B}}

\newcommand{\M}{\mathcal{M}}
\newcommand{\brd}{\Rrightarrow}
\newcommand{\logic}{$\Re$}

\newcommand{\ro}{\mathfrak{r}}

\newtheorem{definition}{Definition}[section]
\newtheorem{theorem}[definition]{Theorem}
\newtheorem{lemma}[definition]{Lemma}

\newtheorem{example}[definition]{Example}

\newtheorem{remark}[definition]{Remark}

\usepackage{iftex}

\ifpdf
  \usepackage{underscore}
  \usepackage[T1]{fontenc}
\else
  \usepackage{breakurl}
\fi

\title{Robust Classification in ML: \\
A Topological Semantics Approach}
\author{ Dominik Pichler
    \institute{TU Wien}
    \email{dominik@logic.at}
\and
Mirko Tagliaferri
    \institute{TU Wien}
    \email{mirko@logic.at} }

 \def\titlerunning{Robust Classification in ML: A Topological Semantics Approach}
 \def\authorrunning{D. Pichler, M. Tagliaferri}

\begin{document}
\maketitle

\begin{abstract}
Robust classification is commonly understood as the stability of a classifier under small perturbations (often adversarial) of input data. In this paper, we propose a logical framework for robust classification grounded in topological semantics for modal logic. Evaluation points are feature vectors representing machine-readable objects, and formulas express explicit classifications. Robustness is interpreted geometrically as local truth persistence: a classification is robust at a point if it holds throughout some non-empty open neighbourhood of that point. Building on this perspective, we introduce a logical language with a robustness modality interpreted over S4 topological spaces, together with a robustness-sensitive conditional connective. This conditional connective captures global inclusion relations between robust regions and other properties of the classifier: it holds at a point when the neighbourhood witnessing the robustness of one formula is contained in the truth set of another. In this way, robust classifications can be systematically linked to classification conditions. We provide a sound and complete axiomatisation of the resulting logic. Finally, we introduce Minimal Robust Models, a constructive method for generating models from specified robustness constraints, which yields formal tools for analysing, explaining, and structuring robust classification behaviour.
\end{abstract}

\section{Introduction}

Among the applications of modal logic, an important one is representing and reasoning about spatial phenomena and the structural properties of spaces. Such an application is made possible thanks to the versatility of modal logic, which is interpretable by a variety of semantic structures, among which topological semantics occupy a prominent position. Introduced in the seminal work of McKinsey and Tarski \cite{McKinseyTarsi44}, topological spaces can be used to interpret modal operators, with the necessity (or $\Box$) operator corresponding to the interior operator. This interpretation offers a geometric reading of the modality, where a formula $\square\varphi$ is true at an evaluation point $s$ whenever $\varphi$ holds throughout some open neighbourhood of that point $s$.

Since \cite{McKinseyTarsi44}, topological semantics has become a central topic in modal logic, connecting spatial reasoning, algebraic semantics and dynamic systems \cite{vanBenthem2007, BlackburnRijkeVenema2001}. In this paper, we aim to continue this tradition by developing a logical framework that exploits the connection between logical modalities and the geometric properties of a space to reason about robust classification in machine learning \cite{Freiesleben2023-yj}.

Machine learning classifiers typically operate by mapping computational representations of objects to discrete labels. These representations interpret objects as points in a high-dimensional feature space, where classification decisions depend on the geometric arrangement of those points. However, it is known that classifiers can be susceptible to input perturbations, in which slight changes to an image or data point can lead to radically different predictions, a phenomenon widely discussed in the literature on adversarial machine learning \cite{8294186,Hosseini18,Meng22}. From a geometric perspective, such classification failures occur when the classifier's decision boundary is too close to the object's representation boundary, making classification unstable to small perturbations in the input.

Topological semantics provides a natural formal language for expressing a notion of local stability of classifications. If feature vectors are interpreted as points of a topological space, then neighbourhoods correspond to perturbations of the underlying object representation. Under this interpretation, the modal operator $\square$ captures the idea that a classification holds throughout some neighbourhood of an evaluation point (a specific vector representation), and thus remains stable under perturbations.

Building on this idea, we propose a modal language in which formulas express explicit classifications and feature properties of objects represented by the vectors. Robust classification is then captured by a robust modality with an added conditional connective that relates robust classification to further classification properties of a given object. Intuitively, the connective will express that the neighbourhood witnessing the robustness of one property lies entirely within the region satisfying another property.

The paper is structured as follows: in section \ref{sec: mot}, we provide motivations behind our proposal as well as providing a first introduction to our methodological approach; in section \ref{sec: rel}, we discuss related works; in section \ref{sec: interp}, we discuss the main interpretation elements of our logic; in section \ref{sec: logic}, we introduce our logic; in section \ref{section: robust models}, we discuss how to construct minimal models; in section \ref{sec: example}, we provide some examples of how we expect our logic to be used; conclusions and future works follow.

\section{Motivation and Methodology}\label{sec: mot}
Machine learning systems classify complex objects (such as images or text) by transforming them into feature vectors, i.e., tuples of measurable properties. These vectors form a feature space in which classifiers learn decision boundaries between categories \cite{Bishop2006}. Formally, a classifier can be viewed as a function $f: X\rightarrow C$, where $X$ is the feature vector space and $C$ is a set of classification labels. Feature vectors represent the model's representation of objects, while classification labels correspond to the externally visible categories.

This distinction is important for explainability in machine learning. A key aim of explainable AI is to relate classification outputs to interpretable properties of inputs. Our goal is to provide a modal logic framework that formally expresses and analyses such relationships, as well as an explicit procedure which guides the construction of the semantics used to interpret the formulas of our language. In our approach, formulas represent external classifications, while semantics are defined over feature vectors. Using topological semantics, robustness of classification can be interpreted as neighbourhood-based stability in the feature space.

This notion of robust classification is a central concern in machine learning. Ideally, objects with slightly different representations should receive the same classification. However, many models exhibit fragile behaviour where small perturbations in input lead to different predictions. We interpret robustness geometrically as local invariance: a classification is robust at a feature vector if it holds throughout a neighbourhood around that vector.

Formally, our logic includes a modal operator representing robust classification, $\R$, and introduces a conditional connective, $\brd$, that links robust classification regions with other properties of the feature space. This framework provides formal tools for analysing and structuring explanations of classifier behaviour while highlighting structural properties of robust classification.

Our methodological approach combines: (i) a clear modelling assumption about how robustness arises from neighbourhood stability, (ii) a training strategy that uses refined labels to induce robust and non‑robust regions, thus providing indications on how the logic's semantic structures might be obtained, and (iii) a formal language for analysing the resulting structure.

\section{Related Works}\label{sec: rel}

The idea of intersecting topology with modal logic is not new \cite{McKinseyTarsi44}. This approach sparked very interesting results, both technical and conceptual, allowing the formalisation of various notions. A clear example of this is \cite{Baltag2022}, where the authors used topological spaces and interiors to represent epistemic concepts such as knowledge and beliefs. Specifically, their topological semantics is used to represent evidence, justifications and, ultimately, beliefs. Moreover, like our approach, they include a conditional belief operator to reason about belief dependencies. Beyond the conceptual differences between our approach and theirs (we model ML classification, where they model evidence and beliefs), from a technical point of view, our conditional connective generalises the one they propose. In particular, their conditional belief operator will hold in our models whenever classifications are assumed to be always robust (i.e., there are no borderline cases). Moreover, we provide a direct completeness proof for the axiomatisation of the language in which the robustness-sensitive conditional is taken as primitive.

Several other logics discuss the use of topologically defined modal operators, including universal modalities to reason about global properties of spaces (akin to our $\brd$) \cite{GorankoPassy1992,Shehtman99,Kontchakov08}. These logics are as expressive as our logic, but, as we will describe in detail in section \ref{sec:similarities}, our contribution is to isolate a robustness‑oriented conditional and to connect it to machine‑learning practice. We also differ from general region‑based spatial logics, such as the Region Connection Calculus \cite{RandellCuiCohn92}, which reason about adjacency and overlap of regions but do not distinguish between stable and unstable instances within a region.

From an interpretative point of view, to the best of our knowledge, our approach is similar to two research papers. In \cite{Kawamoto2021-di}, the author proposes a modal logic for the formal specification of various statistical properties of supervised machine learning models. These properties include classification performance (e.g., precision, recall, accuracy), robustness against adversarial inputs, and fairness notions (e.g., independence, separation, sufficiency). In this logic, formulas represent properties of machine learning models, while the semantics' states are interpreted as datasets, with accessibility relations between states representing transformations applied to those datasets. This approach is similar in spirit to ours, in which data representations dominate the logic's semantics, while syntax is devoted to representing data features. However, unlike the cited work, for us, states will stand for vector representations of single data points rather than full datasets. Moreover, instead of focusing on statistical transformations of datasets, we will focus on geometrical transformations of the feature vectors.

The second paper, \cite{Hornischer2026}, should be considered our main source of inspiration, as it employs the same core ideas we present here. In the paper, a modal logic framework is proposed to analyse the notions of robustness and trustworthiness of AI models from a formal epistemological perspective. The language proposed by the author is used to express the behavioural properties of an AI model, and its semantics is defined over states representing possible inputs to the model (in a more abstract way than ours, but basically aligning with our proposal). The modal operator then captures robustness conditions relating local robustness properties and a more global notion of trustworthy behaviour. While this approach, like ours, uses modal logic to reason about robustness in machine learning systems, the interpretative focus differs. \cite{Hornischer2026}'s framework aims to analyse epistemological constraints on robustness and trustworthiness of model behaviour, while our goal is to provide a sound and complete logic for reasoning about robust classification in the geometric structure of feature spaces. Moreover, although the author presents a topological semantics that captures a non-uniform notion of robustness equivalent to ours, he does so by relying on Euclidean topologies, thereby obtaining trivialisation results that label too many machine learning models as providing robust classifications. In this sense, in our paper, we follow one of his hints, pointing to the idea that similarity between images based on human cognition has an inherently different structure from norm-based structures. We will propose a procedure based on human cognition to foster the creation of classification spaces aligned with human perceptions of robustness.

Finally, our work complements a substantial body of research on adversarial robustness in machine learning. Adversarial examples demonstrate that many classifiers are sensitive to small perturbations in the input \cite{Goodfellow15}; training methods such as adversarial training \cite{Madry18} aim to reduce this vulnerability, and verification tools aim to certify robustness \cite{Katz17,Huang17,Seshia22}. These approaches typically compute, or approximate robust regions post‑training, and do not provide a language for specifying or reasoning about their structure. Our logic fills this gap by offering a formal apparatus for describing robust regions, stating containment relations, and reasoning about classification behaviour under perturbations.

\section{Interpretation}\label{sec: interp}
We will now briefly explain our interpretation of the formal elements introduced in the next section.

Machine learning classifiers operate on feature-based representations of objects, where a feature is a measurable attribute used for prediction (e.g., pixel values in an image, word occurrences in a document, or sensor readings). Let $F$ be a finite set of features, where each feature $f\in F$ has a domain of possible values $D_f$. Assigning values to all features yields a feature vector $s$ that represents a possible object.\footnote{Although the term 'feature vector' can be used (especially in deep learning) to refer to any representation given by the machine learning system, including representation in inner layers, we will use the term as a synonym to 'input vector', thus referring to the initial representation of the object under analysis.} The set of all such vectors forms a feature vector space $X$, which is the semantic domain of our models.

Classification labels correspond to the categories predicted by the model (e.g., cat, dog, spam). In our logic, such basic classification properties are represented by atomic formulas  $p, q, \dots$. The interpretation of an atom consists of all feature vectors for which the classification holds. More complex formulas describe combinations or relationships between classification properties. Then, the satisfaction relation $s\models p$ will indicate that the feature vector $s$ is classified as $p$.

To capture robustness, we add a topological structure to the feature vector space. Intuitively, two vectors are close if they differ only slightly in their feature values, where we rely on an intuitive notion of slight difference, rather than a numerical threshold.\footnote{Formalising such a concept would require a paper in itself, and we will leave it for future work.} Neighbourhoods represent small perturbations of object representations. For each atomic classification $p$, we consider a robust classification set $\ro(p)$, consisting of feature vectors whose classification remains unchanged under small perturbations. The choice of topology depends on the application domain and may often be derived from an underlying metric on the feature space, as suggested in \cite{Hornischer2026}. In this paper, however, we propose a mechanism that could enable the machine learning system to derive the topology autonomously by first using a better supervised training dataset and then building the minimal topology based on the learned robustness constraints (see section \ref{section: robust models}). The main idea behind this approach will be that a modeller might be interested that the classifier displays robust classification behaviour over some specific input data (e.g., it might judge that clear stop signs must be recognised as such), but might be willing to accept error over unclear or non-critical data points (e.g., it might not care whether the classifier makes mistakes over parking signs or struggle to recognise signs in extremely poor weather conditions).

Note that our framework abstracts from the classifier's internal architecture. The logic describes the behavioural relationship between feature vectors and classification properties rather than the computational mechanism that produces the classification, allowing the framework to apply to a wide range of models, including neural networks and decision trees.

\section{Logic}\label{sec: logic}

\subsection{Syntax and Semantics}

The language $\mathcal{L}$ of \logic ~ consists of a countable set $At$ of atomic formulas or atoms (ranging over $p, q, \dots$), the connectives $\wedge$ and $\neg$ of classical logic (the other propositional connectives are defined as usual), the binary operator $\brd$, and the unary operator $\R$ for robustness. $\mathcal{L}$ is defined by the following grammar:
\begin{equation*}
\varphi ::=  p \mid \neg \varphi \mid \varphi \wedge \varphi \mid \R(\varphi) \mid \varphi \brd \varphi
\end{equation*}

The formula $\R(\varphi)$ is read as ``$\varphi$ robustly holds,'' and $\varphi \brd \psi$ is read as ``$\varphi$ robustly implies $\psi$''.

\begin{definition}[Abbreviations]
\[
\bot:= (\varphi\land\neg\varphi) \quad \text{and} \quad \top:= \neg\bot \quad \text{and} \quad \Box \varphi := \top \brd \varphi \quad \text{and} \quad \lozenge \varphi := \neg \Box \neg \varphi = \neg(\top \brd \neg \varphi)
\]
\end{definition}

As seen later in this paper, the derived operators $\Box$ and $\lozenge$ will behave like global \textbf{S5}-modalities in \logic.

\medskip

We now define the pre-models, called \emph{labelled vector space}, which consist of a non-empty set of states (our feature vectors) and a valuation function.\footnote{We will use the term ``state'' instead of ``feature vector'' to align with the standard terminology of modal logic.}

\begin{definition}[Pre-semantics]
Let $X\not=\emptyset$ be a feature vector space (i.e., a set of feature vectors), and let $v: At \to \mathcal{P}(X)$ be a valuation function. The pair $(X, v)$ is called a \emph{labelled vector space}. For $s \in X$, the satisfaction relation $X,s\models \varphi$ (to be read as ``the classification property $\varphi$ holds for the feature vector $s$ in the feature vector space $X$'') is defined inductively:
\begin{align*}
X, s \models p &\iff s \in v(p) \\
X, s \models \neg \varphi &\iff X, s \not\models \varphi \\
X, s \models \varphi \wedge \psi &\iff X, s \models \varphi \text{ and } X, s \models \psi
\end{align*}
\end{definition}

To capture robust classification, we introduce a topological structure into our semantics.

\begin{definition}[Robust Classification Space]
A pair $(X, \Tau)$ is a \emph{robust classification space} if $\Tau \subseteq \mathcal{P}(X)$ satisfies:
\begin{itemize}
    \item[$\text{(i)}$] $\emptyset, X \in \Tau$
    \item[$\text{(ii)}$] $O_1, O_2 \in \Tau \Rightarrow O_1 \cap O_2 \in \Tau$
    \item[$\text{(iii)}$] $\{O_i\}_{i \in I} \subseteq \Tau \Rightarrow \bigcup_{i \in I} O_i \in \Tau$
\end{itemize}
The elements of $\Tau$ are called \emph{robust classification sets}.\footnote{It should be noted that we are not introducing anything new at this stage. We used the terms robust classification space and robust classification sets only to align with our terminology, but these are traditional topological spaces and open sets.}
\end{definition}

Given a labelled vector space $(X, v)$ and a classification property $\varphi$, we want the feature vectors for which the classification $\varphi$ is robust to be those that are surrounded by other feature vectors classified as $\varphi$. At this stage, we assume that the regions of the feature vector space in which robust classifications occur have already been specified (although see subsection \ref{section: robust models} for a description on how to construct them). Thus, we assume a robust classification space has already been defined and that we know which sets constitute robust classification sets. Using the available robust classification space, we can represent the sets of feature vectors for which a formula $\varphi$ is robustly classified via the interior operator.

\begin{definition}[Interior]
Let $(X, \Tau)$ be a robust classification space. For any $A \subseteq X$, the \emph{interior} of $A$ is:
\[
\mathrm{int}(A) := \bigcup \{ O \in \Tau : O \subseteq A \}
\]
That is, $\mathrm{int}(A)$ is the largest robust classification set contained within $A$. 
\end{definition}

    We will use the following well-known properties of the interior operator throughout this paper.

\begin{remark}[Properties of the Interior Operator]

Let $(X, \Tau)$ be a robust classification space. For all $A, B \subseteq X$: $\mathrm{int}(A) \subseteq A$, $\mathrm{int}(\mathrm{int}(A)) = \mathrm{int}(A)$, $\mathrm{int}(A \cap B) = \mathrm{int}(A) \cap \mathrm{int}(B)$ and $A \subseteq B$ implies $\mathrm{int}(A) \subseteq \mathrm{int}(B)$. 

\end{remark}

A state $s$ is contained in the interior of a set $A$ if there exists a robust classification set $O\in \Tau$ such that $s \in O \subseteq A$. In other words, if a state is contained in the interior of the set of all $\varphi$ states, then it is surrounded by only $\varphi$ states.

We are now going to define the rest of our semantics.

\begin{definition}[Model]
Let $(X, \Tau)$ be a robust classification space and $v: At \to \mathcal{P}(X)$ a valuation function. We define a \emph{(classification) model} as $\mathcal{M} := (X, \Tau, v)$.
\end{definition}

\begin{definition}\label{def:semantics}(Full Semantics)
Let $\mathcal{M} = (X, \Tau, v)$ be a model and $s \in X$. Our satisfaction relation is defined as follows, where $||\varphi||^{\mathcal{M}} := \{ s \in X : \M, s \models \varphi \}$:
\begin{align*}
\M, s &\models p \iff s \in v(p) \\
\M, s &\models \neg \varphi \iff \M, s \not\models \varphi \\
\M, s &\models \varphi \wedge \psi \iff \M, s \models \varphi \text{ and } \M, s \models \psi \\
\M, s &\models \R \varphi \iff s \in \mathrm{int}(|| \varphi ||^{\mathcal{M}}) \\
\M, s &\models \varphi \brd \psi \iff \emptyset \neq \mathrm{int}(|| \varphi ||^{\mathcal{M}}) \subseteq || \psi ||^{\mathcal{M}}
\end{align*}

\end{definition}

Semantic consequence (\(\Phi \models \varphi\)) and validity (\(\models \varphi\)) are defined as usual. \\

Before delving further into the properties of our models, we provide some informal intuition and observations. Given a classification $\psi$, we can see that $\M, s \models \psi$ in general does not imply $\M, s \models \R\psi$. In such a case, $s$ is found at the so-called classification border of $|| \psi ||^{\mathcal{M}}$. Being a borderline case means that any robust classification set $O \in \Tau$ containing $s$ intersects $||\neg \psi ||^{\mathcal{M}}$, in other words, $s$ is very close to not being classified as a $\psi$. Although the classifier classified the feature vector $s$ as a $\psi$, such classification was not robust. For example, if we take $\psi$ to mean ``it is a 5'', and the feature vector $s$ represents Figure \ref{fig:MNIST5},\footnote{Figure \ref{fig:MNIST5} is taken from the MNIST database, which is a well-known database used to train machine learning systems to classify handwritten digits. Note also that, although in our logic such an image would be represented as a feature vector rather than a 2-dimensional image, the vector representation of the image would not be easily interpretable by a human, so we choose to present the example with the given figure instead.} we might be inclined to say that although the classification is correct, even small changes to such an image might modify the classification, e.g., making it be classified as a 3.

\begin{figure}[h]
    \centering
    \includegraphics[scale=0.3]{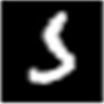}
    \caption{MNIST example of the digit 5.}
    \label{fig:MNIST5}
\end{figure}

Concerning our $\brd$ connective, we are interested in seeing what classification information can be extracted from the fact that a robust classification already took place. For example, imagine that $\varphi$ stands for ``it is an 8'', while $\chi$ stands for ``it contains two clearly defined circles'', we would want to say that any robust classification of $\varphi$ implies $\chi$ ($\varphi\brd\chi$). In other words, any feature vector that is robustly classified as an 8 should also validate the property that there are two clearly defined circles. However, $\varphi\rightarrow\chi$ might still not hold, since there might be borderline classification of the feature vector as an 8 that does not contain two clearly defined circles (e.g., one of the circles might not be fully closed).

Now we take a closer technical look at how the semantics behave. First, the operator $\R$ is an \textbf{S4}-modality and satisfies axioms \textbf{K}, \textbf{T}, and \textbf{4}, as shown in prior work on topological semantics \cite{McKinseyTarsi44}. Next, the derived operator $\Box$ is a global \textbf{S5}-modality as seen in the next lemma. Later in the paper, we also see that the axioms \textbf{K}, \textbf{T}, and \textbf{5} are derivable for $\Box$ in the associated proof system.

\begin{lemma}\label{lem: universal}
    Given a model $\M=(X,\Tau,v)$ and a state $s \in X$:

    \begin{enumerate}
        \item $\M, s \models \Box \varphi$ iff $\forall s\in X: \M, s \models \varphi$
        \item $\M, s \models \lozenge \varphi$ iff $\exists s\in X: \M, s \models \varphi$
    \end{enumerate}

\end{lemma}

\begin{proof}
1) If $\M, s \models \Box \varphi$ holds then by definition $int(||\top||^\M) \subseteq ||\varphi||^\M$. Since $int(||\top||^\M)=int(X)=X$ we get that $||\varphi||^\M = X$. The other direction follows from the fact that $X$ is non-empty. 2) The claim for $\lozenge$ follows directly from 1).
\end{proof}

\subsection{The semantics under the microscope}\label{sec:similarities}

In this section, we discuss the similarities and differences between our semantics and other approaches in topological semantics and modal logic.

Our language contains the \textbf{S4}-modality $\R$, interpreted as the interior operator, and the connective $\brd$, both as primitive operators. Lemma~\ref{lem: universal} shows that, in our semantics, the universal \textbf{S5}-modality is definable by $\Box \varphi := \top \brd \varphi$. Conversely, if we replace $\brd$ with $\Box$ as a primitive operator, then $\brd$ can be expressed using the topological modality $\R$ together with the universal modality $\Box$ as follows:
$\varphi \brd \psi \equiv \lozenge \R \varphi \wedge \Box(\R \varphi \to \psi)$.
The first conjunct, $\lozenge \R \varphi$, says that the robust region of $\varphi$ is non-empty. The second conjunct, $\Box(\R \varphi \to \psi)$, says that every state belonging to the robust region of $\varphi$ satisfies $\psi$. Since
$||\R\varphi||^\M = int(||\varphi||^\M)$, this is equivalent to requiring that
$\emptyset \neq int(||\varphi||^\M) \subseteq ||\psi||^\M$,
which is exactly the satisfaction condition for $\varphi \brd \psi$. Thus, over non-empty topological models, the language with $\R$ and $\brd$ is definitionally equivalent to the standard topological language of \textbf{S4} equipped with the universal modality \cite{Lenzen1978,GorankoPassy1992}. The role of $\brd$ is therefore not to increase the expressive power of topological \textbf{S4} with the universal modality, but to isolate a robustness-sensitive pattern that is central for the intended application. The formula $\varphi \brd \psi$ packages two conditions: first, that $\varphi$ has a non-empty robust region, and second, that this robust region is globally included in the truth set of $\psi$. In the context of robust classification, this means that whenever $\varphi$ holds robustly at some point, every robust instance of $\varphi$ also satisfies the classification property $\psi$.

The non-emptiness condition is important. Without this semantic constraint, the operator would be equivalent to $\Box(\R\varphi \to \psi)$, which means that every formula with an empty interior would robustly imply every formula, because the inclusion
$int(||\varphi||^\M) \subseteq ||\psi||^\M$
would hold vacuously whenever
$int(||\varphi||^\M)=\emptyset$.
The conjunct $\lozenge \R\varphi$ prevents this vacuity. Thus $\varphi \brd \psi$ does not merely say that all robust $\varphi$-states are $\psi$-states; it also says that there exists at least one robust $\varphi$-state. Although inclusion-based readings of implication are standard in modal and topological semantics, this explicit non-vacuity requirement plays a distinctive role in the present setting, since it makes the connective sensitive to whether the antecedent has any genuine robust instances. This additional requirement also blocks the ``blow-up'' of material implication. Indeed, the negation $\neg(\bot \brd \varphi)$ is valid, since $int(\emptyset)=\emptyset$. As a consequence, the operator $\brd$ in general does \textbf{not} allow for strengthening of the antecedent: even if $\varphi \brd \psi$ holds and $\chi \to \varphi$ holds throughout the model, it need not follow that $\chi \brd \psi$ holds. More specifically, strengthening of the antecedent fails when the strengthened antecedent has an empty interior. In this sense, the non-empty-interior condition marks a point at which $\brd$ departs from pure containment conditionals and behaves more like a robustness-sensitive conditional: a conditional statement is licensed only when the antecedent has a non-empty region of stable verification.

\begin{example}
    Take a labelled vector space $(X,v)$ with $X=\{s_1, s_2\}$, $v(q)=\{s_1, s_2\}$ and $v(p)=\{s_1\}$. The model $\M:=(X,\Tau,v)$ consists of a trivial robust classification space $\Tau=\{\emptyset, X\}$. Then $\M, s_1 \models \Box (p \rightarrow q)$ and $\M, s_1 \models q \brd q$, but $\M, s_1 \not\models p \brd q$.
\end{example}

The operator $\brd$ validates several familiar cumulative principles from non-monotonic consequence relations and preference-based conditionals, where principles such as Cut and Cautious Monotonicity play a central role
\cite{Gabbay1985,KrausLehmannMagidor1990,Aqvist1987,Parent2015}. In particular, $\brd$ validates Right Weakening, And, Cautious Cut, and Cautious Monotonicity:

\begin{itemize}
    \item \textbf{Right Weakening / Weakening of the Consequent:} If $\models \psi \to \chi$, then
    $\models (\varphi \brd \psi) \to (\varphi \brd \chi)$.

    \item \textbf{Conjunction of Consequents:}
    $\models ((\varphi \brd \psi) \land (\varphi \brd \chi)) \to
    (\varphi \brd (\psi \land \chi))$.

    \item \textbf{Cautious Cut:}
    $\models ((\varphi \brd \psi) \land ((\varphi \wedge \psi) \brd \chi))
    \to (\varphi \brd \chi)$.

    \item \textbf{Cautious Monotonicity:}
    $\models ((\varphi \brd \psi) \land (\varphi \brd \chi))
    \to ((\varphi \land \psi)\brd \chi)$.
\end{itemize}

At the same time, the non-empty-interior requirement makes $\brd$ more restrictive than standard KLM-style consequence relations. As a consequence of the restriction placed on the antecedent, $\brd$ does not behave as a
monotonic conditional: it does not allow for arbitrary strengthening of the antecedent. Moreover, $\brd$ in general does not validate
\textbf{Identity}, i.e. $\varphi \brd \varphi$\footnote{This principle is
also known as Reflexivity.}, whenever
$int(||\varphi||^\M)=\emptyset$. Instead, it validates a
weaker non-vacuous form of Identity, namely
$\models \lozenge \R \varphi \to (\varphi \brd \varphi)$. In fact, the
converse also holds, so that
$\models (\varphi \brd \varphi) \leftrightarrow \lozenge \R \varphi$.
Thus, Identity is recovered exactly for formulas whose truth set has a
non-empty interior, i.e., for formulas with at least one robust instance.

\subsection{Axiomatisation}\label{sec:axio}

In this subsection, we present a complete set of axioms for our logic \logic ~ and demonstrate its soundness and completeness with respect to our classification models.

\subsubsection{Axioms}

\begin{align*}
\text{(CL)} \quad & \text{All instances of classical propositional logic} \\
\text{(K-R)} \quad & \R(\varphi \to \psi) \to (\R \varphi \to \R \psi) \\
\text{(4)} \quad & \R\varphi \to \R \R \varphi \\
\text{(K-RC)} \quad & (\varphi \brd (\psi \to \chi)) \to ((\varphi \brd \psi) \rightarrow (\varphi \brd \chi)) \\
\text{(Abs)} \quad & \lozenge (\varphi \brd \psi) \to \Box (\varphi \brd \psi) \\
\text{(RMP)} \quad & (\R \varphi \land (\varphi \brd \psi)) \to \psi \\
\text{(Dis)} \quad & (\varphi \brd \psi) \to (\varphi \brd \R \psi) \\
\text{(Mon)} \quad & \lozenge \R(\varphi \land \chi) \to ((\varphi \brd \psi) \to (\varphi \land \chi \brd \psi)) \\
\text{(NV)} \quad & (\varphi \brd \psi) \to \lozenge \R (\varphi)
\end{align*}

\subsubsection{Inference Rules}

\[
\textbf{MP}: \frac{\varphi, \quad \varphi \to \psi}{\psi}
\qquad
\textbf{RCI}: \frac{\varphi \to \psi}{\lozenge \R\varphi \to (\varphi \brd \psi)}
\qquad
\textbf{N}_{\R}: \frac{\varphi}{\R \varphi}
\]

\begin{definition}\label{def:derivation}  
A derivation of \(\varphi_n\) is a sequence \(\varphi_1, ..., \varphi_n\) where each \(\varphi_i\) is either an axiom instance or follows from the previous ones by applying one of the rules. We write \(\Gamma \vdash \varphi\) if there is a derivation of \(\varphi\) or if \(\psi_1 \land ... \land \psi_n \rightarrow \varphi\) for some \(\psi_1, ..., \psi_n \in \Gamma\).  
\end{definition}

We first note that $\Box$ and $\lozenge$ fulfil all the axioms of an S5 modality as well as the rule of necessitation for $\Box$. 

\begin{lemma}\label{lemma:boxdiamond}
The S5 axioms \textbf{K}, \textbf{T} and \textbf{5} for $\Box$ and $\lozenge$ are derivable in this system, as well as the necessitation rule for $\Box$.
\end{lemma}

\begin{proof}
    \textbf{K-Axiom: $\Box(\varphi \to \psi) \to (\Box \varphi \to \Box \psi) $}

\begin{enumerate}
    \item $\vdash ((\top \brd (\varphi \to \psi)) \to ((\top \brd \varphi) \to (\top \brd \psi)))$ \hfill by \textbf{K-RC}
    \item $\vdash \Box (\varphi \to \psi) \to (\Box \varphi \to \Box \psi)$ \hfill Definition of $\Box$
\end{enumerate}

\textbf{T-Axiom: $\Box \varphi \to \varphi$}

\begin{enumerate}
    \item $\vdash (\R \top \land (\top \brd \varphi)) \to \varphi$ \hfill by \textbf{RMP}
    \item $\vdash \R \top$ \hfill by \textbf{N}$_\R$ from classical tautology $\top$
    \item $\vdash (\top \brd \varphi) \to \varphi$ \hfill from (1) and (2)
    \item $\vdash \Box \varphi \to \varphi$ \hfill definition of $\Box$
\end{enumerate}

\textbf{5-Axiom: $\lozenge \Box \varphi \to \Box \varphi$}

\begin{enumerate}
    \item $\vdash \lozenge(\top \brd \varphi) \to \Box (\top \brd \varphi)$ \hfill by \textbf{Abs}
    \item $\vdash \Box (\top \brd \varphi) \to (\top \brd \varphi)$ \hfill by T-Axiom (from above)
    \item $\vdash \lozenge \Box \varphi \to \Box \varphi$ \hfill from (1) and (2) and definition of $\Box$
\end{enumerate}

\textbf{Necessitation for $\Box$}

\begin{enumerate}
    \item $\vdash \varphi$ \hfill Assumed to be derivable
    \item $\vdash \top \to \varphi$ \hfill by \textbf{CL}
    \item $\vdash \lozenge \R \top \to (\top \brd \varphi)$ \hfill by \textbf{RCI} on (2)
    \item $\vdash \R \top$ \hfill by \textbf{N}$_\R$ as before
    \item $\vdash (\top \brd \varphi)$ \hfill from (3), (4) and \textbf{T}
    \item $\vdash \Box \varphi$ \hfill definition of $\Box$
\end{enumerate}
\end{proof} 

As a next step, we show that our axiomatisation of \logic~ is strongly sound with regard to the class of all classification models.

\begin{theorem}[Strong Soundness]\label{teo:soundness} 
 If  $\Pi \vdash \varphi$, then $\Pi \models \varphi$.
\end{theorem}

\begin{proof}
    We proceed by induction on the derivation length, distinguishing cases according to the last rule applied. We show below the details for axioms \textbf{Mon}, \textbf{RMP} and \textbf{Dis} and rule $\mathbf{RCI}$. The other rules and axioms are left up to the reader.

\textbf{Mon}: Given $\M=(X,\Tau,v)$ and $s \in X$ such that $\M,s \models \lozenge \R(\varphi \land \chi) \land (\varphi \brd \psi)$, This means that there exists a state $w \in X$ such that $w \in int(||\varphi \land \chi||^\M)$ and $int(||\varphi||^\M) \subseteq ||\psi||^\M$. Putting those two together and by the definition of the interior operator, we get that $w \in int(||\varphi \land \chi||^\M) \subseteq int(||\varphi ||^\M) \subseteq ||\psi||^\M$. This means that $int(||\varphi \land \chi||^\M) \not= \emptyset$ and $int(||\varphi \land \chi||^\M) \subseteq ||\psi||^\M$, hence $\M,s \models \varphi \land \chi \brd \psi$. 

\textbf{RMP}: Given $\M=(X,\Tau,v)$ and $s \in X$ such that $\M,s \models \R \varphi \land (\varphi \brd \psi$), then we get $s \in int(||\varphi||^\M) \subseteq ||\psi||^\M$. This means $\M,s \models \psi$.

\textbf{Dis}: Given $\M=(X,\Tau,v)$ and $s \in X$ such that $\M,s \models \varphi \brd \psi$, then we get $int(||\varphi||^\M) \subseteq ||\psi||^\M$. By the monotonicity property  of the interior operator we derive $int(int(||\varphi||^\M)) \subseteq int(||\psi||^\M)$ and furthermore $int(||\varphi||^\M) \subseteq int(||\psi||^\M)$ by the property $int(int(||\varphi||^\M)=int(||\varphi||^\M)$. Hence, $\M, s \models \varphi \brd \R(\psi)$ .

\textbf{RCI}: Given  $\M=(X,\Tau,v)$  with $\M \models \varphi \rightarrow \psi$. First, we take a state $s \in X$ such that $\M,s \models \lozenge\R \varphi$. Hence $int(||\varphi||^\M) \not= \emptyset$ and by the assumptions $\M \models \varphi \rightarrow \psi$ we also get $int(||\varphi||^\M) \subseteq ||\psi||^\M$. By the property  $int(||\varphi||^\M) \subseteq ||\varphi||^\M$ of the interior operator we derive $\emptyset \not= int(||\varphi||^\M) \subseteq ||\psi||^\M$, which means $\M, s \models \varphi \brd \psi$. 
\end{proof}

To prove completeness, we construct a canonical model through maximal consistent sets and show that the logic is sound and complete with regard to arbitrary sets $X$. Since the language is definitionally equivalent to topological \textbf{S4} with the universal modality, completeness for the present language could also be obtained indirectly by translating formulas with $\brd$ into the language with $\R$ and $\Box$. Nevertheless, we give a direct canonical proof for the present axiomatisation. The purpose of the proof is not to introduce a new canonical-model technique, but to establish the adequacy of the proof system for the language in which $\brd$ is taken as primitive. In particular, the proof makes explicit how the axioms and rules capture the two semantic components of $\brd$: the non-emptiness of the robust antecedent region and its global inclusion in the consequent.

\begin{definition}\label{def:MCS}
    A set $\Gamma \subseteq\mathcal{L}$ is called a maximal consistent set (MCS for short) if (a) $\Gamma \not\vdash \bot$, and (b) for every $\varphi \in\mathcal{L}$ either $\varphi \in \Gamma$ or $\neg \varphi \in \Gamma$. We write $MCS$ for the set of all maximal consistent sets.
\end{definition}

Since the operators $\brd$ and the derived modality $\Box$ behave as global modalities, we need to make sure that all the states inside the canonical model satisfy the same $\brd$ formulas. Hence, we can not use all MCS for the canonical model. Instead, we will use a set of MCS constructed from one initial MCS $\Theta$. 

\begin{definition}
    Given a MCS $\Gamma$ we define $\Gamma^{-1} := \{ \varphi \in \mathcal{L}: \top \brd \varphi \in \Gamma \}$. Given a fixed MCS $\Theta$ we define $\Omega := \Theta^{-1}$ and $X_\Omega := \{ \Gamma \in \text{MCS} : \Gamma^{-1} = \Omega \}$.
\end{definition}

The set $X_\Omega$ contains all MCS that share the same $\Box \varphi$ formulas as $\Theta$. 

As a next step, we need to define the robust classification space $\Tau$ on our canonical model. For that, we are going to define a set $\B \subseteq P(X_\Omega)$ which will serve as the base of $\Tau$. A robust classification base $\B$ of a robust classification space $\Tau$ consists of robust classification sets such that every such set of the robust classification space can be represented as the union of some subfamily of $\B$. Our base consists of the following sets:

\begin{definition}
Given a formula $\varphi \in\mathcal{L}$ then $O_\varphi := \{ \Gamma \in X_\Omega : \R \varphi \in \Gamma \}$
\end{definition}

To make sure that a set $\B \subseteq P(X_\Omega)$ generates a robust classification space, we need to verify two properties. First, we need to ensure that $\B$ generates the whole set $X_\Omega$; we call this property $B_1$. Secondly, we need to make sure that $\B$ is closed under intersection; we call this property $B_2$. \footnote{Although property $B_2$ is more specific than the general condition typically required in the literature for topological bases, we will still satisfy it.}

The following verification follows the standard topo-canonical model construction for \textbf{S4}. The basic open sets are determined by formulas whose robust versions belong to a maximal consistent set. We include the argument only to keep the construction self-contained and to make clear how the topology is generated in the present language.

\begin{lemma}
    The set $\mathcal{B} := \{ O_\varphi \subseteq X_\Omega : \varphi \in \mathcal{L} \}$ fulfils properties $B_1$ and $B_2$.
\end{lemma}

\begin{itemize}
   \item $B_1$: This follows from the fact that $O_\top = X_\Omega$ since $\R \top$ is contained in every MCS by the rule $\textbf{N}_\R$.
    \item $B_2$: Given $\Gamma \in O_\varphi \cap O_\psi$, then we have $\R(\varphi) \in \Gamma$ and $\R(\psi) \in \Gamma$. Using the classical propositional tautology $\varphi \rightarrow (\psi \rightarrow (\varphi \land \psi))$ the rule $\textbf{N}_\R$ and the axiom \textbf{K-R} we derive $(\R(\varphi) \land \R(\psi)) \rightarrow \R(\varphi \land \psi)$. Since maximal consistent sets are closed under consequence, we can conclude  $\Gamma \in O_{\varphi \land \psi}$. This means $ O_\varphi \cap O_\psi \subseteq  O_{\varphi \land \psi}$. Similarly, we can also derive $\R(\varphi \land \psi) \to (\R(\varphi) \land \R(\psi))$ which lets us conclude $O_\varphi \cap O_\psi \supseteq  O_{\varphi \land \psi}$. Consequently, $\B$ is closed under intersection.
\end{itemize}

 The main additional point, compared with the ordinary canonical topological model for \textbf{S4}, concerns the connective $\brd$. Since $\brd$ behaves globally, similarly to the universal modality, the canonical model is built over $X_\Omega$, the class of maximal consistent sets sharing the same global theory. The next lemma records the canonical condition corresponding to $\varphi \brd \psi$: the robust antecedent must be non-empty in the model, and its robust region must be included in the truth set of the consequent.

\begin{lemma}\label{lem:Equiv}
Let $\Gamma \in X_\Omega$. Then the following are equivalent:
\[
\begin{array}{lcl}
1.~\varphi \brd \psi \in \Gamma & & 2.~\R \varphi \to \psi \in \Omega \text{~ and ~} \neg \R \varphi \notin \Omega
\end{array}
\]
\end{lemma}

\begin{proof}
(1 $\Rightarrow$ 2): Applying the rule of necessitation for $\Box$ and the axiom \textbf{K} for $\Box$ to the axiom  \textbf{RMP} $\varphi \brd \psi \rightarrow (\R\varphi \rightarrow \psi)$ we derive  $\Box(\varphi \brd \psi) \rightarrow \Box(\R\varphi \rightarrow \psi)$. Through the axiom \textbf{Abs} and the assumption $\varphi \brd \psi \in \Gamma$ we get $\top \brd (\R \varphi \to \psi) \in \Gamma$, hence $\R \varphi \to \psi \in \Omega$. From $\varphi\brd \psi \in \Gamma$ we get $\lozenge \R \varphi \in \Gamma$ through \textbf{NV}, hence $\neg \R \varphi \not\in \Omega$.

(2 $\Rightarrow$ 1): By the assumption $\R \varphi \to \psi \in \Omega$ we have $\top \brd (\R \varphi \to \psi) \in \Gamma$ and since $\neg \R \varphi \not\in \Omega$ we have $\Box \neg \R \varphi \not\in \Gamma$. Hence, by maximality $\lozenge \R \varphi \in \Gamma$. Through the axiom \textbf{Mon} we get $\lozenge \R \varphi \to [(\top \brd (\R \varphi \to \psi)) \to (\varphi \brd (\R \varphi \to \psi))] \in \Gamma$. Since $\Gamma$ is an MCS, we have $\varphi \brd (\R \varphi \to \psi) \in \Gamma$. Applying \textbf{K-RC} we arrive at $(\varphi \brd \R \varphi) \to (\varphi \brd \psi) \in \Gamma$. Since $(\varphi \brd \varphi) \in \Gamma$ because of the rule of \textbf{RCI} and $\lozenge \R \varphi \in \Gamma$ it follows that $(\varphi \brd \R \varphi) \in \Gamma$ from the axiom \textbf{Dis}. Finally, we arrive at $\varphi \brd \psi \in \Gamma$.
\end{proof}

\begin{remark}
    Lemma \ref{lem:Equiv} shows that the axiomatisation of \logic ~ syntactically recovers the semantic non-vacuity condition of the conditional operator $\brd$ through the axiom \textbf{NV}.
\end{remark}

\begin{lemma}[Witness Lemma for $\lozenge$]\label{Lemma:Wittness}
Let $\Gamma \in X_\Omega$ and let $\chi \in \mathcal L$. Then $\lozenge \chi \in \Gamma$ iff $\text{there exists } \Delta \in X_\Omega \text{ such that } \chi \in \Delta$.
\end{lemma}

\begin{proof}[Proof sketch]
For the left-to-right direction, assume $\lozenge\chi\in\Gamma$. We show that $\{\chi\}\cup\{\Box\alpha:\alpha\in\Omega\}\cup\{\neg\Box\alpha:\alpha\notin\Omega\}$ is consistent. If not, then by the finitary definition of $\vdash$, some finite subset would already be inconsistent; using necessitation and the S5 principles for $\Box$, this would yield $\Box\neg\chi\in\Gamma$, contradicting $\lozenge\chi\in\Gamma$. Hence, the set is consistent, and, by the Lindenbaum Lemma, it can be extended to a maximal consistent set $\Delta$. By construction, $\chi\in\Delta$ and $\Delta^{-1}=\Omega$, so $\Delta\in X_\Omega$.

For the right-to-left direction, assume that there is some $\Delta\in X_\Omega$ with $\chi\in\Delta$. If $\lozenge\chi\notin\Gamma$, then by maximality $\Box\neg\chi\in\Gamma$. Since $\Gamma^{-1}=\Omega$, this gives $\neg\chi\in\Omega$. But $\Delta^{-1}=\Omega$, so $\Box\neg\chi\in\Delta$; by the  \textbf{T} axiom for $\Box$, we obtain $\neg\chi\in\Delta$, contradicting $\chi\in\Delta$. Therefore $\lozenge\chi\in\Gamma$.
\end{proof}

\begin{lemma}[Truth Lemma]\label{lem:truth-lemma}
Let $\M^{Can} = (X_\Omega, \Tau, v)$ be where $X_\Omega$ and  $\Tau$ are defined as before and $v: At \to P(X_\Omega)$ is defined as $v(p):=\{\Gamma \in X_\Omega: p \in \Gamma\}$. Then for every $\Gamma \in X_\Omega$ and $\varphi \in \mathcal{L}$ the following equivalence holds: $\varphi \in \Gamma \iff \M^{Can}, \Gamma \models \varphi$
\end{lemma}

\begin{proof}
The proof is done by induction on the formula construction of $\varphi$. We will only discuss the cases $\varphi \brd \psi$ and $\R(\varphi)$. The other cases are handled as usual and left to the reader.

\textbf{Case $\R \varphi$:} $\R\varphi\in \Gamma$ implies that $\Gamma \in O_\varphi$. By definition $O_\varphi \in \Tau$ and furthermore $O_\varphi \subseteq ||\varphi||^{\M^{Can}}$ by the induction hypothesis. Hence $\Gamma \in int(||\varphi||^{\M^{Can}})$ and therefore $\M^{Can}, \Gamma \models \R\varphi$.

$\R\varphi \not\in \Gamma$ implies that $\Gamma \not\in O_\varphi$. Let us assume there exists a $O \in \B$ such that $\Gamma \in O \subseteq ||\varphi||^{\M^{Can}}$. This means there exists a $\psi \in\mathcal{L}$ with $O_\psi \subseteq ||\varphi||^{\M^{Can}}$. Hence, every maximal consistent set in $X_\Omega$ contains $\R(\psi) \rightarrow \varphi$. As a consequence, we can find a finite number of $\chi_1, ..., \chi_n \in \Theta$ such that $\vdash (\Box \chi_1\land ... \land \Box \chi_n) \rightarrow (\R(\psi) \rightarrow \varphi)$. Applying \textbf{N}$_\R$, \textbf{Abs} and \textbf{4} we get $\vdash (\Box \chi_1\land ... \Box \chi_n) \rightarrow (\R(\psi) \rightarrow \R(\varphi))$. This means $\R(\varphi)$ is contained in every MCS of $X_\Omega$, which is in $O_\psi$, contradicting the assumption $\R\varphi \not\in \Gamma$. Finally, we derive that there can not exist a set in $\B$ such that $\Gamma \in O \subseteq ||\varphi||^{\M^{Can}}$, which means $\Gamma \not\in int(||\varphi||^{\M^{Can}})$ and therefore $\M^{Can}, \Gamma \not\models \R\varphi$.

\textbf{Case $\varphi \brd \psi$:} $\varphi \brd \psi \in \Gamma$ directly implies $\M^{Can}, \Gamma \models \varphi \brd \psi$ by 1. $\Rightarrow$ 2. of Lemma \ref{lem:Equiv} and the induction hypothesis.

For $\varphi \brd \psi \not\in \Gamma$, we discuss the two cases following from 2. $\Rightarrow$ 1. of Lemma \ref{lem:Equiv}. First case: $\R(\varphi) \rightarrow \psi \not\in \Omega$. We have $\Box(\R\varphi\to\psi)\notin\Gamma$, hence $\lozenge\neg(\R\varphi\to\psi)\in\Gamma$ by maximality. By Lemma \ref{Lemma:Wittness} there exists $\Delta\in X_\Omega$ such that
$\neg(\R\varphi\to\psi)\in\Delta$. Thus $\R\varphi\in\Delta$ and
$\psi\notin\Delta$. Therefore the $\R$-case of the truth lemma together with the induction hypotheses for $\varphi$ and $\psi$ gives $\Delta \in int(||\varphi||^{\M^{Can}})$ and $\Delta\notin||\psi||^{\M^{Can}}$, hence $int(||\varphi||^{\M^{Can}})\nsubseteq||\psi||^{\M^{Can}}$, so $\M^{Can},\Gamma\not\models\varphi\brd\psi$. Second case: $\neg \R(\varphi) \in \Omega$. Then by the induction hypothesis we derive $int(||\varphi||^{\M^{Can}}) = \emptyset$ which also means $\M^{Can}, \Gamma\not\models \varphi \brd \psi$.
\end{proof}

\begin{theorem}[Strong Completeness]\label{teo:completness}
If $\Pi \models \varphi$ then $\Pi \vdash \varphi$.
\end{theorem}

\begin{proof}
We argue by contraposition. Assume $\Pi \not\vdash \varphi$. Then the set $\Pi \cup \{\neg \varphi\}$ is consistent and can therefore be extended to a maximal consistent set $\Delta$. Let $\Omega := \Delta^{-1}$ and consider the canonical model $\M^{Can} = (X_\Omega,\Tau,v)$ based on $\Omega$. By Lemma~\ref{lem:truth-lemma}, every formula in $\Delta$ is true at $\Delta$ in $\M^{Can}$. Hence $\forall \psi \in \Pi:\ \M^{Can},\Delta \models \psi$ and $\M^{Can},\Delta \models \neg \varphi$. So $\M^{Can},\Delta \not\models \varphi$, and therefore $\Pi \not\models \varphi$ in the canonical model.
\end{proof}

The preceding completeness theorem is a completeness theorem with respect to arbitrary topological spaces. Since the language is countable, the canonical model has cardinality at most continuum, and its underlying set could therefore be injected into some subset of $\mathbb{R}^n$. However, this observation should not be understood as a completeness result for Euclidean feature spaces equipped with their usual metric topology. The topology transported from the canonical model will, in general, not coincide with the usual Euclidean topology. Establishing completeness or representation results for more restricted classes of spaces, such as metric spaces or Euclidean spaces with their standard topology, is a separate question which we leave for future work when explicit metrics are studied.

\section{Minimal robust models}\label{section: robust models}

In this section, we explain how a modeller might construct the topological space based on her cognitive abilities. Our approach to robust classification combines a practical heuristic for learning robust regions with a mathematical procedure that constructs a minimal topological space from specified robust subsets.

\subsection{First step: heuristic}

The first step is based on the idea that, given an atomic formula $p \in At$, representing a basic classification property and a labelled vector space $(X,v)$, the modeller may already have intuitions about which subset of $v(p)$ should be robustly classified as $p$. In many applications, each feature vector corresponds to a concrete object, allowing the modeller to identify clear instances of a class that should remain invariant under small perturbations, thereby providing stronger supervision to the classifier, e.g., by providing extra classificatory labels for the robust objects. Such human-centred classification should only ensure nesting between simple and robust classifications.

We denote by $\ro(p) \subseteq v(p)$ the set of such robust $p$-vectors. Semantically, this reflects the requirement that small perturbations of the vectors representing the objects in $\ro(p)$ do not change their classification. In contrast, vectors outside $\ro(p)$ may lie near the boundary of the classification region and be more sensitive to perturbations. From a model-theoretic perspective, this corresponds to requiring that $\ro(p) \subseteq \mathrm{int}(v(p))$, as specified in the previous paragraph.

In practice, however, the robust region associated with a class is not limited to the explicitly identified examples, but includes all vectors in a neighbourhood of each robust instance. Since these neighbourhoods are typically high-dimensional and continuous, they cannot be exhaustively specified. To address this, we adopt a multi-label heuristic: the modeller labels certain training instances as robust (e.g., robust $p$) and others as standard $p$ instances. A classifier trained on this enriched dataset learns two nested regions—a broader region corresponding to $v(p)$ and a smaller core approximating $\ro(p)$.

For example, in Figure ~\ref{fig:MNIST8}, the modeller may specify that $s_1, s_2$, and $s_3$ are robustly classified as eights. At the same time, $s_4$ and $s_5$ are labelled as standard eighths, reflecting the modeller's ability to distinguish clear instances from ambiguous ones, without requiring precise knowledge of the entire robust region. This approach, inspired by adversarial training in machine learning \cite{Madry18}, provides a practical approximation of robust classification sets while avoiding the need to specify them exhaustively.

Once these robust subsets have been identified, they can serve as a basis for constructing the model's topological structure.

\begin{figure}[h]
    \centering
    \includegraphics[scale=0.3]{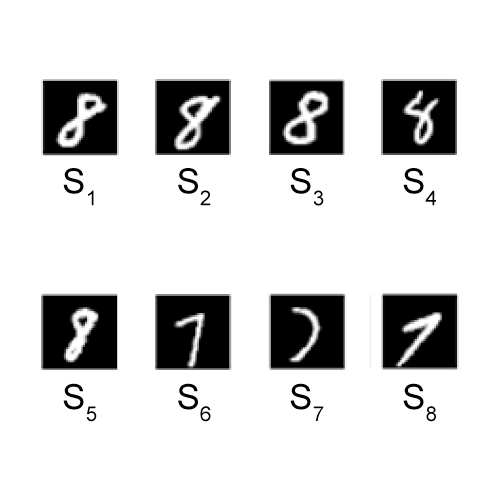}
    \caption{MNIST example of the digits 7 and 8.}
    \label{fig:MNIST8}
\end{figure}

\subsection{Second step: topology extraction}
Given a labelled vector space $(X,v)$, and a function $r: At \to P(X)$ with $\ro(p) \subseteq v(p)$ (as constructed through the previous heuristic), one can define multiple classification models for which a given set $\ro(p) \subseteq int(v(p))$. For example, we could always choose a discrete robust classification space $P(X)$ on $X$, because then the classification set of every formula $\varphi$ would be equivalent to the classification set of $\R(\varphi)$, i.e., the robust classification set. So in general, there exists a whole class of such models.

The construction in this section is mathematically straightforward: given a specified family of robust atomic regions, it generates the smallest topology containing them. Its purpose is not to introduce a new topological construction, but to give a modelling principle for robust classification. If a modeller specifies which atomic classifications count as robust (e.g., by checking the validation and test performance of an initial classification), the minimal robust model provides the least topological structure required to validate those robustness assumptions. In this sense, the construction avoids adding robustness beyond what is forced by the specified atomic robust regions and by closure under the topological operations.

\begin{definition}
Given a labelled vector space $\mathcal{D}:=(X,v)$ and a function $\ro:At\to \mathcal{P}(X)$ with $\ro(p)\subseteq v(p)$, we write $\mathbb{M}^{\mathcal{D},\ro}$
for the class of all models $M=(X,\Tau,v)$ such that, for every $p\in At$, $\ro(p)\in \Tau$.\footnote{Since $\ro(p)\in\Tau$ and $\ro(p)\subseteq v(p)$, it follows that
$\ro(p)\subseteq int(v(p))$.}
\end{definition}

Given the interpretation of a formula $\R(\varphi)$ in a model, we see that the models in $\mathbb{M}^{\mathcal{D},\ro}$ agree on the robust classification of $p$ in $\ro(p)$. Given that the discrete robust classification model $(X, P(X), v)$ is always part of $\mathbb{M}^{\mathcal{D},\ro}$, we know that there always exists at least one robust classification model for any given $\mathcal{D}$ and $r$. The problem with the discrete robust classification space is that the model using it considers every formula as robustly classified, since $\R(\varphi)$ and $\varphi$ collapse in such a model.

We want to work with a robust classification model that uses the fewest robust classification sets while still treating every feature vector in $\ro(p)$ as a robust $p$-vector. To find such a model, we order the models in $\mathbb{M}^{\mathcal{D},\ro}$ via the following pre-order, for which we write $\M_1 \leq \M_2$ if $\M_1$ considers fewer formulas robust than $\M_2$. More formally:

\begin{definition}\label{def:ordering}
 Given a labelled vector space $(X,v)$. We call a model $\M_1=(X, \Tau_1, v)$ at least as strict as a model  $\M_2=(X, \Tau_2, v)$, written as $\M_1 \leq \M_2$ iff for every propositional formula $\varphi$ and every state $s \in X$, if $\M_1, s \models \R(\varphi)$, then $\M_2, s \models \R(\varphi)$.
\end{definition}

The question that naturally arises is whether there exists a smallest (strictest) model in $\mathbb{M}^{\mathcal{D},\ro}$. Hence, we are going to show that such a smallest model $\M=(X, \Tau, v)$ always exists within $\mathbb{M}^{\D, \ro}$. Meaning that for any model $\M'=(X, \Tau', v) \in \mathbb{M}^{\D, \ro}$ we have $\M \leq \M'$.

\begin{theorem}
    There exists a smallest model $\M=(X, \Tau, v)$ in $\mathbb{M}^{\D, r}$ with respect to the relation $\leq$ of Definition \ref{def:ordering}. 
\end{theorem}

\begin{proof}[Proof sketch]
We define the robust classification space $\Tau$ of the smallest model $\M$ as the one built from $\ro(p)$. This means the set $\mathcal{S}:=\{ \ro(p): p \in At\}$ serves as the sub-base of $\Tau$. As such, a set $O \subseteq X$ is an element of $\Tau$ iff it can be written as the union of finite intersections of elements in $\mathcal{S}$.\footnote{The empty intersection is defined as the whole set $X$.} By definition, every element in $\mathcal{S}$ is an element of the robust classification space. Hence we get $\ro(p) \in \Tau$ which implies that $\ro(p)$ is a robust classification set contained in $v(p)$, making $\ro(p)$ subset of $\mathrm{int}(v(p))$ and therefore $(X, \Tau, v) \in \mathbb{M}^{\D, r}$. Furthermore, $\Tau$ is the smallest robust classification space containing the sets in $\mathcal{S}$, hence every robust classification set in $\Tau$ is contained in every robust classification space of a model in $(X, \Tau, v) \in \mathbb{M}^{\D, r}$. As a consequence, given a state $s \in X$ and a propositional formula $\varphi$, if  $(X, \Tau, v), s \models \R(\varphi)$ holds, then there exists a robust classification set $O \in \Tau$ such that $s \in O \subseteq ||\varphi||^\M$. Since $\varphi$ is propositional, its truth set is determined only by the valuation $v$, and therefore $||\varphi||^\M = ||\varphi||^{\M'}$. Since $\Tau$ is the smallest robust classification space containing $\mathcal{S}$ we can conclude that $O$ is an element of every robust classification space in the class $\mathbb{M}^{\D, r}$ giving us $\M', s \models \R(\varphi)$ for any model $\M' \in \mathbb{M}^{\D, r}$.
\end{proof}

In summary, our methodology combines a data‑driven heuristic that uses multi‑labelling and training to approximate robust subsets with a mathematical construction that formalises these subsets as open sets in a minimal topology. The heuristic addresses the practical difficulty of identifying robust regions. At the same time, the mathematical procedure guarantees that the formal semantics align with the modeller's robustness assumptions without imposing unnecessary robustness on other instances.

\section{Examples}\label{sec: example}

In this section, we will discuss some examples to explain the various parts of our logic.

In this first example, we will model a simple image classification scenario. The model is meant to try to classify numbers. Through \logic, it is possible to represent various properties of the classifications being made. The formulas considered in the example will be: $8=$ ``The image was classified as an 8''; $7=$ ``The image was classified as a 7''; $2\text{-}circles=$ ``The image was classified as having two clear circles''; $StraightLine=$ ``The image was classified as containing a straight line''.\footnote{Note that given the nature of the example, the classification categories for the numbers are mutually exclusive, but this is not a pre-requisite for our logic, as can be seen with the properties $2\text{-}circles$ and $StraightLine$.}

The model will receive eight images as input, as shown in Figure \ref{fig:MNIST8}. The feature vectors representing those images are indicated with $s_i$ with $i \in \{1, \dots, 8\}$ as shown in the figure. We also assume that a human has already selected the robust classification sets during image preprocessing. Specifically, the following classifications were given: $v(8)=\{s_1, s_2, s_3, s_4, s_5\}$; $v(7)=\{s_6, s_7, s_8\}$; $v(2\text{-}circles)=\{s_1, s_2, s_3\}$; $v(StraightLine)=\{s_1, s_2, s_6, s_8\}$.

For the robust classification sets, there is only one which relates to the classification of the 8s. In particular, we assume that $\ro(8)=\{s_1, s_2, s_3\}$, which means that the robust classification space will be\footnote{Obviously, the more robust classification sets are included in the formalisation, the richer the robust classification space will be, but to keep the example simple, we chose to specify just one over and above the empty and whole set.}:
\[
\Tau=\{
\emptyset,
\{s_1, s_2, s_3\},
\{s_1, s_2, s_3, s_4, s_5, s_6, s_7, s_8\}
\}
\]
Note that the interior of $8$ is $int(8)=\ro(8)$, as expected. Moreover, although the modeller might not have specified it, it turns out that the interior of $2\text{-}circles$ is also $int(2\text{-}circles)=\{s_1, s_2,s_3\}=v(2\text{-}circles)$, which means that all the classifications of images indicating 2-circles are robust. Finally, both $7$ and $StraightLine$ have an empty interior, meaning that it is not possible to provide a robust classification boundary for those properties. In this setting, $8\rightarrow 2\text{-}circles$ is not true at every state of the model (it does not hold, e.g., in $s_4$), but $8\brd 2\text{-}circles$ is, which follows from the fact that borderline classification cases of images of 8s do not have the property of having two circles. However, when only robust classifications are considered, such a property emerges, which means that if we can establish the robustness of an 8-classification, we are guaranteed to obtain the 2-circles property.

We will now highlight what happens when the satisfiability conditions for $\varphi\brd \psi$ fail to hold in a model. First, we show what it means that the interior of the antecedent $\varphi$ is empty. To show this, we concentrate on the images of 7s and the formula $StraightLine$. In this case, the interior of $7$ is $\emptyset$, i.e., for any image classified as a 7, even slight changes in the feature vector representing it could drastically alter its classification. In this scenario, although we expect images of 7s to contain straight lines, it is reasonable to expect that the classification model cannot learn any meaningful connections between 7s and straight lines, since all available inputs are poor images. Thus, the formula $7\brd StraightLine$ does not hold.

Finally, we will show what happens when the consequent of a robust implication fails to hold within the boundaries of the interior of the antecedent. To do so, note that although $8$ has an interior, it is not a subset of $StraightLine$, since image $s_3$, which is robustly classified as an 8, does not contain a straight line, implying the failure of $8\brd StraightLine$, since, as we expect, having a robust classification of an 8 should not guarantee that the image also contains a straight line.

\section{Conclusion and Future Works}
We introduced a logical framework for reasoning about the robustness of machine learning classifications using topological semantics. We included in the framework a conditional connective that can be used to derive new information from features that are known to be robust. Our formal framework has the potential to improve understanding of robustness in ML by highlighting logical connections between classification properties. In the future, we aim to extend this work in the following directions: first, we would like to study the computational complexity of the model-checking and satisfiability problems to gauge how feasible an implementation of the logic would be; second, we would like to study the minimal robust models proposed in section \ref{section: robust models} in more detail, by investigating the logic behind the robustness of certain features and their effect on the axiomatisation proposed in section \ref{sec:axio}; third, we would like to add a measure over the topology to allow reasoning about \textit{how much} a classification is robust; fourth, we plan to integrate trust evaluations and probability values within the formalism, similar to how it has been done in \cite{Aldini2021,Tagliaferri24,TagliaferriAldini2022} and following the footsteps of \cite{Hornischer2026}.

\paragraph{Acknowledgements:} Work supported by the FWF project LoDEx Grant-DOI 10.55776/I6372 and the FWF project Exploring Conditional Logics via Proof Theory Grant-DOI 10.55776/PAT1815025.

\nocite{*}
\bibliographystyle{eptcs}
\bibliography{generic}

\end{document}